\documentclass[11pt,letterpaper]{article}
\usepackage{amsmath,amsfonts,amsthm,amssymb,tabls,stmaryrd,graphicx,relsize}

\theoremstyle{plain}

\numberwithin{equation}{section}
\newtheorem{thm}{Theorem}[section]

{\begin{flushleft}\textbf{Example #1}.\enspace}%
{\end{flushleft}}

\newcommand{\complex}{{\mathbb C}}
\newcommand{\positive}{{\mathbb N}}
\newcommand{\real}{{\mathbb R}}
\newcommand{\ascript}{{\mathcal A}}
\newcommand{\bscript}{{\mathcal B}}
\newcommand{\cscript}{{\mathcal C}}
\newcommand{\dscript}{{\mathcal D}}
\newcommand{\pscript}{{\mathcal P}}
\newcommand{\rscript}{{\mathcal R}}
\newcommand{\qscript}{{\mathcal Q}}
\newcommand{\tscript}{{\mathcal T}}

\newcommand{\rmcyl}{\mathrm{cyl}}
\newcommand{\rmtr}{\mathrm{tr}}
\newcommand{\rmre}{\mathrm{Re\,}}
\newcommand{\ftilde}{\widetilde{f}}
\newcommand{\omegahat}{\widehat{\omega}}
\newcommand{\offspring}{\!\shortrightarrow\,}

\newcommand{\ab}[1]{\left|#1\right|}
\newcommand{\brac}[1]{\left\{#1\right\}}
\newcommand{\paren}[1]{\left(#1\right)}
\newcommand{\sqbrac}[1]{\left[#1\right]}
\newcommand{\elbows}[1]{{\left\langle#1\right\rangle}}
\newcommand{\ket}[1]{{\left|#1\right>}}
\newcommand{\bra}[1]{{\left<#1\right|}}

\begin{document}

\title{CAUSAL SET APPROACH TO\\DISCRETE QUANTUM GRAVITY
}
\author{S. Gudder\\ Department of Mathematics\\
University of Denver\\ Denver, Colorado 80208, U.S.A.\\
sgudder@du.edu
}
\date{}
\maketitle

\begin{abstract}
We begin by describing a sequential growth model in which the universe grows one element at a time in discrete time steps.
At each step, the process has the form of a causal set and the ``completed'' universe is given by a path consisting of a discretely growing chain of causal sets. We then introduce a quantum dynamics to obtain a quantum sequential growth process (QSGP) which may lead to a viable model for discrete quantum gravity. A discrete version of Einstein's field equation is derived and a definition for discrete geodesics is proposed. A type of QSGP called an amplitude process is introduced. An example of an amplitude process called a complex percolation process is studied. This process conforms with general principles of causality and covariance. We end with some detailed quantum measure calculations for a specific percolation constant.
\end{abstract}

\section{Introduction}  
The causal set approach to discrete quantum gravity is an attempt to unify general relativity and quantum mechanics
\cite{blms87, sor03, sur11}. These two theories are quite different and it is not at all clear how such a unification is possible. This question has been investigated for about 80 years and is probably the greatest unsolved problem in theoretical physics. Briefly speaking, quantum theory is based on the study of self-adjoint and unitary operators on a complex Hilbert space, while general relativity is based on the study of smooth curves and tensors on a $4$-dimensional, real, differentiable manifold $M$ with a Lorentzian metric tensor $g_{\mu\nu}$.

We first examine $M$ more closely. For each $a\in M$ there is a forward light cone $C_a^+\subseteq M$ consisting of points in the future of $a$ that $a$ can communicate with via a light signal. If $b\in C_a^+$ we say that $b$ is in the
\textit{causal future} of $a$ and write $a<b$. Then $(M,<)$ becomes a \textit{partially ordered set} (poset); that is, $a\not<a$ (irreflexivity) and $a<b$, $b<c$ imply that $a<c$ (transitivity). We call $(M,<)$ the \textit{causal structure} on $M$. To remind us that we are dealing with causal structures, we call an arbitrary finite poset a \textit{causal set} (or \textit{causet}). Investigators have shown that the causal structure completely determines $M$ \cite{sor03, sur11}. That is, $<$ determines the topology, differential structure, smooth functions, dimension, line and volume elements and tensor $g_{\mu\nu}$ for $M$. We can therefore forget about the differential structure of $M$ and only consider the poset $(M,<)$ which is clearly a great simplification.

In comparing quantum mechanics (especially quantum field theory) and general relativity, we notice one important similarity. They both contain many singularities and the theories break down at small distances. This indicates that quantum mechanics should be based upon a finite-dimensional complex Hilbert space and that general relativity should be discrete with a minimum distance which we take to be a Planck length $\ell _p\approx 1.6\times 10^{-33}$ cm. and a minimum time which we take to be a Planck instant $t_p\approx 5.4\times 10^{-44}$ sec.

Beside discreteness, a second motivating feature of the causal set approach is that the universe is expanding both in size and in matter creation. These features suggest that we should consider a discrete sequential growth model. In such a model, the universe grows one element at a time in discrete steps given by Planck instants. At each step, the universe has the form of a causet and the ``completed'' universe is given by a path consisting of a discretely growing chain of causets. We then introduce a quantum dynamics $\rho _n$, $n=1,2,\ldots$, to obtain a quantum sequential growth process (QSGP) which may lead to a viable model for discrete quantum gravity. The dynamics $\rho _n$ is given by a positive operator on the Hilbert space of causet paths of length $n$. The operators $\rho _n$ are required to satisfy certain normalization and consistency conditions.

At this stage of development, the precise form of $\rho _n$ is not known. However, a discrete version of Einstein's field equation is derived and it is possible that $\rho _n$ can be specified by determining whether this discrete equation is approximated by the classical Einstein equation. We also propose a definition for discrete geodesics.

Although various constructions of a QSGP $\rho _n$ are known \cite{gud111, gud112, gud121}, we now introduce a particularly simple type called an amplitude process. An example of an amplitude process called a complex percolation process is studied. This process conforms with general principles of causality and covariance. Some detailed quantum measure calculations for a specific percolation constant are performed and some geodesics are briefly examined.

\section{Sequential Growth Model} 
Let $\pscript _n$ be the collection of all causets of cardinality $n$, $n=1,2,\ldots$, and let $\pscript =\cup\pscript _n$. Two isomorphic causets are considered to be identical. If $x\in\pscript$ and $a,b\in x$, we say that $a$ is an \textit{ancestor} of
$b$ and $b$ is a \textit{successor} of $a$ if $a<b$. We say that $a$ is a \textit{parent} of $b$ and $b$ is a \textit{child} of $a$
if $a<b$ and there is no $c\in x$ with $a<c<b$. We call $a$ \textit{maximal} in $x$ if there is no $b\in x$ with $a<b$.
If $x\in\pscript _n$, $y\in\pscript _{n+1}$, then $x$ \textit{produces} $y$ (and $y$ is a \textit{product} of $x$) if $y$ is obtained from $x$ by adjoining a single element $a$ to $x$ that is maximal in $y$. Thus, $a$ is not in the causal past of any element of $y$. If $x$ produces $y$, we write $x\to y$. The transitive closure of $\offspring$ makes $\pscript$ into a poset and we call $(\pscript ,\offspring )$ a \textit{sequential growth model}. A \textit{path} in $\pscript$ is a string (sequence)
 $\omega =\omega _1\omega _2\cdots$, $\omega _i\in\pscript _i$ and $\omega _i\to\omega _{i+1}$, $i=1,2,\ldots\,$. An
 $n$-\textit{path} is a finite string $\omega =\omega _1\omega _2\cdots\omega _n$, where again $\omega _i\in\pscript _i$ and
 $\omega _i\to\omega _{i+1}$. We denote the set of paths by $\Omega$ and the set of $n$-paths by $\Omega _n$. If $x$ produces $y$ in $r$ isomorphic ways, we say that the \textit{multiplicity} of $x\to y$ is $r$ and write $m(x\to y)=r$. For example, in Figure~1, $m(x_3\to x_6)=2$ and multiplicities greater than $1$ are designated. To be precise, the different isomorphic ways requires a labeling of the causets. This is the only place we need to mention labeled causets and we otherwise only consider unlabeled causets.
 
 We think of a path $\omega\in\Omega$ as a possible universe (universe history) \cite{gud111, gud112, gud121}. For
 $\omega =\omega _1\omega _2\cdots$, $\omega _i\in\pscript _i$ represents a universe at Planck instant $i$. This gives a growth model for the universe \cite{rs00, sur11, vr06}. The vertices of the causet $\omega _i$ represent a space-time framework (scaffolding) at step $i$ (instant $i$). A vertex may or may not be occupied by a point mass or energy. Figure~1 gives the first four steps of the sequential growth model representing possible universes. The vertical rectangles represent antimatter causets, the horizontal rectangles represent matter causets and the circles represent mixed causets. For a discussion of these types, we refer the reader to \cite{gud121}.

\includegraphics*[trim= 0 0 30 90, scale=.75, angle=90]{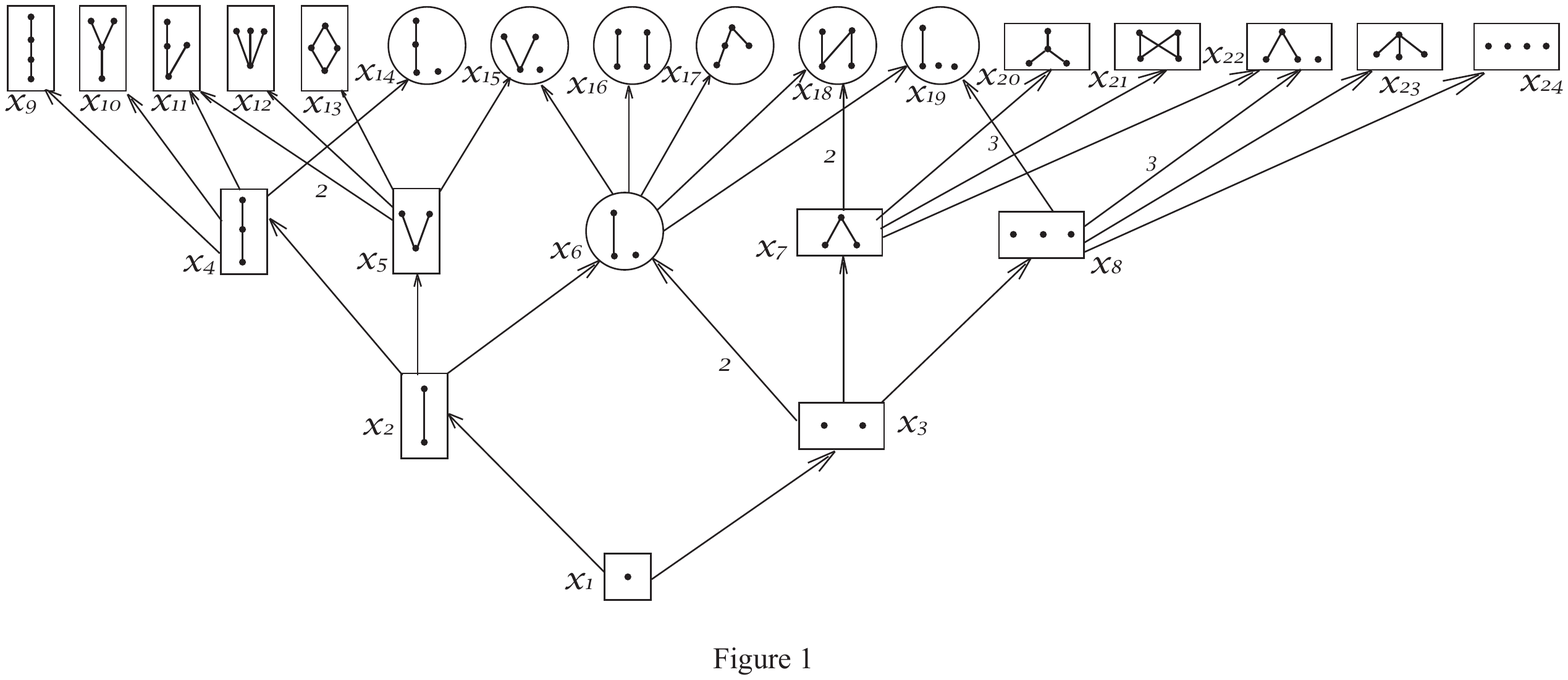}

We use the notation $\ascript _n$ for the power set $2^{\Omega _n}$, $n=1,2,\ldots\,$. For $x\in\pscript _i$ we use the notation
 \begin{equation*}
x\offspring =\brac{y\in\pscript _{i+1}\colon x\to y}
\end{equation*}
and for $\omega =\omega _1\omega _2\cdots\omega _n\in\Omega _n$ we write
\begin{equation*}
\omega\offspring =\brac{\omega _1\omega _2\cdots\omega _n\omega _{n+1}\colon\omega _n\to\omega _{n+1}}
  \in\ascript _{n+1}
\end{equation*}
Finally, for $A\in\ascript _n$ we define
\begin{equation*}
A\offspring =\bigcup _{\omega\in A}(\omega\offspring)\in\ascript _{n+1}
\end{equation*}
The set of paths beginning with $\omega =\omega _1\cdots\omega _n\in\Omega _n$ is called an
\textit{elementary cylinder set} and is denoted $\rmcyl (\omega )$. If $A\in\ascript _n$, then the \textit{cylinder set} $\rmcyl (A)$ is defined by
\begin{equation*}
\rmcyl (A)=\bigcup _{\omega\in A}\rmcyl (\omega )
\end{equation*}
Using the notation
\begin{equation*}
\cscript (\Omega _n)=\brac{\rmcyl (A)\colon A\in\ascript _n}
\end{equation*}
notice that if $A\in\cscript (\Omega _n)$, then $A=\rmcyl (A_1)$ for some $A_1\in\ascript _n$ so
$A=\rmcyl (A_1\offspring)\in\cscript (\Omega _{n+1})$. We conclude that
\begin{equation*}
\cscript (\Omega  _1)\subseteq\cscript (\Omega _2)\subseteq\cdots
\end{equation*}
is an increasing sequence of subalgebras of the \textit{cylinder algebra} $\cscript (\Omega )=\cup\cscript (\Omega _n)$. For
$A\in 2^\Omega$ we define the set $A^n\in\ascript _n$ by
\begin{equation*}
A^n=\brac{\omega _1\omega _2\cdots\omega _n\in\Omega _n\colon\omega _1\omega _2
  \cdots\omega _n\omega _{n+1}\cdots\in A}
\end{equation*}
We think of $A^n$ as the step-$n$ approximation to $A$. Notice that $A^n$ is the set of $n$-paths whose
continuations are in $A$.

\section{Quantum Sequential Growth Processes} 
Denoting the cardinality of a set $A$ by $\ab{A}$, we define the $\ab{\Omega _n}$-dimensional complex Hilbert space $H_n=L_2(\Omega _n,\ascript _n,\nu _n)$ where $\nu _n$ is the counting measure on $\Omega _n$. Of course, $H_n$ is isomorphic to $\complex ^{\ab{\Omega _n}}$. Let $\chi _A$ denote the characteristic function of a set $A\in\ascript _n$ and let $1_n=\chi _{\Omega _n}$. A positive operator $\rho _n$ on $H_n$ satisfying $\elbows{\rho _n1_n,1_n}=1$ is called a
$q$-\textit{probability operator} and the set of $q$-probability operators on $H_n$ is denoted $\qscript (H_n)$. Corresponding to $\rho _n\in\qscript (H_n)$ we have an $n$-\textit{decoherence functional}
$D_n\colon\ascript _n\times\ascript _n\to\complex$ given by
\begin{equation*}
D_n(A,B)=\elbows{\rho _n\chi _B,\chi_A}
\end{equation*}
which gives a measure of the interference between $A$ and $B$. It is easy to show that $D_n$ has the usual properties of a decoherence functional. That is, $D_n(\Omega _n\Omega _n)=1$, $D_n(A,B)=\overline{D_n(B,A)}$, $A\mapsto D_n(A,B)$ is a complex measure on $\ascript _n$ for any $B\in\ascript _n$ and if $A_1,\ldots ,A_m\in\ascript _n$ then $D_n(A_i,A_j)$ are the components of a positive semidefinite $m\times m$ matrix. The map $\mu _n\colon\ascript _n\to\real ^+$ given by
$\mu _n(A)=D_n(A,A)$ is called the $q$-\textit{measure} corresponding to $\rho _n$ \cite{sor94}. We interpret $\mu _n(A)$ as the \textit{propensity} of the event $A$ when the system is described by $\rho _n$ \cite{gud111,gud112}. Notice that
$\mu _n(\Omega _n)=1$. Although $\mu _n$ is not additive, it does satisfy the \textit{grade}-2 \textit{additivity} condition: if $A,B,C\in\ascript _n$ are mutually disjoint, then
\begin{equation*}
\mu _n(A\cup B\cup C)=\mu _n(A\cup B)+\mu _n(A\cup C)+\mu _n(B\cup C)-\mu _n(A)-\mu _n(B)-\mu _n(C)
\end{equation*}

We say that a sequence $\rho _n\in\qscript (H_n)$, $n=1,2,\ldots$, is \textit{consistent} if
$D_{n+1}(A\offspring ,B\offspring )=D_n(A,B)$ for all $A,B\in\ascript _n$. Of course, it follows that
$\mu _{n+1}(A\offspring )=\mu _n(A)$ for all $A\in\ascript _n$. A consistent sequence $\rho _n\in\qscript (H_n)$ provides a quantum dynamics for the growth model $(\pscript ,\offspring )$ and we call $\rho _n$ a
\textit{quantum sequential growth process} (QSGP) \cite{gud111, gud112, gud121}. At this stage of development we do not know the specific form of $\rho _n$ that would describe quantum gravity. It is hoped that further theoretical properties or experimental data will determine $\rho _n$. One possible approach is considered in Section~4.

Let $\rho _n\in\qscript (H_n)$ be a QSGP. Although we have a $q$-measure $\mu _n$ on $\ascript _n$, $n=1,2,\ldots$, it is important to extend $\mu _n$ to physically relevant subsets of $\Omega$ in a systematic way. We say that a set
$A\subseteq\Omega$ is \textit{beneficial} if $\lim\mu _n(A^n)$ exists and is finite in which case we define $\mu (A)$ to be this limit. We denoted the collection of beneficial sets by $\bscript (\rho _n)$. If $A\in\cscript (\Omega )$ is a cylinder set, then
$A\in\cscript (\Omega _i)$ for some $i\in\positive$. In this case $A=\rmcyl (A_1)$ for some $A_1\in\ascript _i$. Now $A^i=A_1$, $A^{i+1}=A_1\offspring$, $A^{i+2}=(A_1\offspring )\offspring$, $\cdots$. Hence,
\begin{equation*}
\lim\mu _n(A^n)=\mu _i(A^i)=\mu _i(A_1)
\end{equation*}
so $A$ is beneficial and $\mu (A)=\mu _i(A_1)$. We conclude that $\cscript (\Omega )\subseteq\bscript (\rho _n)$ and if
$A\in\cscript (\Omega )$ then $\mu (A)=\mu _n(A^n)$ for $n$ sufficiently large. Simple examples are
$\emptyset,\Omega\in\bscript (\rho _n)$ with $\mu (\emptyset )=0$, $\mu (\Omega )=1$. Of course, there are physically relevant subsets of $\Omega$ that are not cylinder sets. For example, if $\omega\in\Omega$ then
$\brac{\omega}\notin\cscript (\Omega )$. Whether $\brac{\omega}\in\bscript (\rho _n)$ depends on $\rho _n$ and we shall consider some examples in Section~6.

\section{Discrete Einstein Equation} 
Let $Q_n=\cup _{i=1}^n\pscript _i$ and let $K_n$ be the Hilbert space $\complex ^{Q_n}$ with the standard inner product
\begin{equation*}
\elbows{f,g}=\sum _{x\in QN}\overline{f(x)}g(x)
\end{equation*}
Let $L_n=K_n\otimes K_n$ which we identify with $\complex ^{Q_n\times Q_n}$. Let $\rho _n\in\qscript (H_n)$ be a QSGP with corresponding decoherence matrices
\begin{equation*}
D_n(\omega ,\omega ')=D_n\paren{\brac{\omega},\brac{\omega '}},\quad\omega ,\omega'\in\Omega _n
\end{equation*}
If $\omega =\omega _1\omega _2\cdots\omega _n\in\Omega _n$ and $\omega _i=x$ for some $i$, then $\omega$
\textit{contains} $x$. For $x,y\in Q_n$ we define
\begin{equation*}
D_n(x,y)=\sum\brac{D_n(\omega ,\omega ')\colon\omega\hbox{ contains§ }x,\omega '\hbox{ contains }y}
\end{equation*}
Due to the consistency of $\rho _n$, $D_n(x,y)$ is independent of $n$ if $n\ge\ab{x},\ab{y}$ where $\ab{x}$ is the cardinality of $x\in Q_n$. Also $D_n(x,y)$, $x,y\in Q_n$, are the components of a positive semi-definite matrix.

We think of $Q_n$ as an analogue of a differentiable manifold and $D_n(x,y)$ as an analogue of a metric tensor. If $y\in Q_n$ and $x\to y$ we think of the pair $(x,y)$ as a tangent vector at $y$. Thus, there are as many tangent vectors at $y$ as there are producers of $y$. Finally, the elements of $K_n$ are analogous to smooth functions on the manifold.

If $\omega =\omega _1\omega _2\cdots\omega _n\in\Omega _n$ and $\omega _i=x$, then $i=\ab{x}$ and $\omega$ contains $x$ if and only if $\omega _{\ab{x}}=x$. An $n$-path $\omega$ containing $x$ determines a tangent vector
$(\omega _{\ab{x}-1},x)$ at $x$ (assuming $\ab{x}\ge 2$). For $\omega\in\Omega _n$ define the \textit{difference operator}
$\varbigtriangleup _\omega ^n$ on $K_n$ by
\begin{equation*}
\varbigtriangleup _\omega ^nf(x)=\sqbrac{f(x)-f(\omega _{\ab{x}-1})}\delta _{x,\omega _{\ab{x}}}
\end{equation*}
where $\delta _{x,\omega _{\ab{x}}}$ is the Kronecker delta. It is easy to check that $\varbigtriangleup _\omega ^n$ satisfies the \textit{discrete Leibnitz rule}:
\begin{equation*}
\varbigtriangleup _\omega ^nfg(x)
  =f(x)\varbigtriangleup _\omega ^ng(x)+g(\omega _{\ab{x}-1})\varbigtriangleup _\omega ^nf(x)
\end{equation*}

Given a function $f\in\complex ^{Q_n\times Q_n}=L_n$ of two variables we have the function $\ftilde\in K_n$ of one variable $\ftilde (x)=f(x,x)$ and given a function $g\in K_n$ we have the functions of two variables $g_1,g_2\in L_n$ where
$g_1(x,y)=g(x)$ and $g_2(x,y)=g(y)$ for all $x,y\in Q_n$. For $\omega ,\omega '\in\Omega _n$, we want
$\varbigtriangleup _{\omega ,\omega '}^n\colon L_n\to L_n$ that extends $\varbigtriangleup _\omega ^n$ and satisfies Leibnitz's rule. That is,
\begin{align}           
\label{eq41}
\varbigtriangleup _{\omega ,\omega '}^ng_1(x,y)
  &=\varbigtriangleup _\omega ^ng(x)\delta _{y,\omega '_{\ab{y}}},
  \varbigtriangleup _{\omega ,\omega '}^ng_2(x,y)
  =\varbigtriangleup _{\omega '}^ng(y)\delta _{x,\omega _{\ab{x}}}\\
\intertext{and}
\label{eq42}
\varbigtriangleup _{\omega ,\omega '}^nfg(x,y)
  &=f(x,y)\varbigtriangleup _{\omega ,\omega '}g(x,y)+g(\omega _{\ab{x}-1},\omega '_{\ab{x}-1})
  \varbigtriangleup _{\omega ,\omega '}^nf(x,y)
\end{align}
The next two theorems are proved in \cite{gud122}

\begin{thm}       
\label{thm41}
A linear operator $\varbigtriangleup _{\omega ,\omega '}^n\colon L_n\to L_n$ satisfies \eqref{eq41} and \eqref{eq42} if and only if it has the form
\begin{equation}         
\label{eq43}
\varbigtriangleup _{\omega ,\omega '}^nf(x,y)=\sqbrac{f(x,y)-f(\omega _{\ab{x}-1},\omega '_{\ab{y}-1})}
  \delta _{x,\omega _{\ab{x}}}\delta _{y,\omega '{\ab{y}}}
\end{equation}
\end{thm}

The result \eqref{eq43} is not surprising because it is the natural extension of $\varbigtriangleup _\omega ^n$ from $K_n$ to $L_n$. Also, $\varbigtriangleup _{\omega ,\omega '}^n$ extends $\varbigtriangleup _\omega ^n$ in the sense that
\begin{equation*}
\varbigtriangleup _{\omega ,\omega}^nf(x,y)=\varbigtriangleup _\omega ^n\ftilde (x)
\end{equation*}

\begin{thm}       
\label{thm42}
{\rm (a)}\enspace A linear operator $T_\omega\colon K_n\to K_n$ satisfies the Leibnitz rule and $T_\omega f(x)=0$ when
$\omega _{\ab{x}}\ne x$ if and only if there exists a function $\beta _\omega\colon Q_n\to\complex$ such that
$T_\omega =\beta _\omega\varbigtriangleup _\omega ^n$.
{\rm (b)}\enspace A linear operator $T_{\omega ,\omega '}\colon L_n\to L_n$ satisfies the Leibnitz rule and
$T_{\omega ,\omega '}f(x,y)=0$ when $\omega _{\ab{x}}\ne x$ or $\omega '_{\ab{y}}\ne y$ if and only if there exists a function $\beta _{\omega ,\omega '}\colon Q_n\times Q_n\to\complex$ such that
$T_{\omega ,\omega '}=\beta _{\omega ,\omega '}\varbigtriangleup _{\omega ,\omega '}^n$.
\end{thm}

It is clear that $\mu _n(x)=D_n(x,x)$ is not stationary. That is, $\varbigtriangleup _\omega ^n\mu _n(x)\ne 0$ for all
$x\in Q_n$ in general. It is shown in \cite{gud122} that the simplest nontrivial combination
$\varbigtriangledown _\omega ^n=\beta _\omega\varbigtriangleup _\omega +\alpha _\omega$ satisfying
$\varbigtriangledown _\omega ^n\mu _n(x)=0$ for all $x\in Q_n$ is given by
\begin{equation*}
\varbigtriangledown _\omega ^nf(x)=\sqbrac{\mu _n(\omega _{\ab{x}-1})f(x)-\mu _n(x)f(\omega _{\ab{x}-1})}
  \delta _{x,\omega _{\ab{x}}}
\end{equation*}
We call $\varbigtriangledown _\omega ^n$ the \textit{covariant difference operator}.

Again, $\varbigtriangleup _{\omega ,\omega '}^nD_n(x,y)\ne 0$ for all $x,y\in Q_n$. It is shown in \cite{gud122} that the simplest nontrivial combination
\begin{equation*}
\varbigtriangledown _{\omega ,\omega '}^n
  =\beta _{\omega ,\omega '}\varbigtriangleup _{\omega ,\omega '}^n+\alpha _{\omega ,\omega '}
\end{equation*}
satisfying $\varbigtriangledown _{\omega ,\omega '}^nD_n(x,y)=0$ for all $x,y\in Q_n$ is given by
\begin{align*}
\varbigtriangledown _{\omega ,\omega '}^nf(x,y)&=\sqbrac{D_n(\omega _{\ab{x}-1},\omega '_{\ab{y}-1})f(x,y)
  -D_n(x,y)f(\omega _{\ab{x}-1},\omega '_{\ab{y}-1})}\\
  &\qquad \delta _{x,\omega _{\ab{x}}}\delta _{y,\omega '_{\ab{y}}}
\end{align*}
We call $\varbigtriangledown _{\omega ,\omega '}^n$ the \textit{covariant bidifference operator}.

The \textit{curvature operator} is defined as
\begin{equation*}
\rscript _{\omega ,\omega '}^n=\varbigtriangledown _{\omega ,\omega '}^n-\varbigtriangledown _{\omega ',\omega}^n
\end{equation*}
We define the \textit{metric operator} $\dscript _{\omega ,\omega '}^n$ on $L_n$ by
\begin{align*}
\dscript _{\omega ,\omega '}^nf(x,y)&=D_n(x,y)\left[f(\omega '_{\ab{x}-1},\omega _{\ab{y}-1})
  \delta _{x,\omega '_{\ab{x}}}\delta _{y,\omega _{\ab{y}}}\right.\\
  &\quad \left.-f(\omega _{\ab{x}-1},\omega '_{\ab{y}-1})\delta _{x,\omega _{\ab{x}}}\delta _{y,\omega '_{\ab{y}}}\right]
\end{align*}
and the \textit{mass-energy operator} $\tscript _{\omega ,\omega '}^n$ on $L_n$ by
\begin{align*}
\tscript _{\omega ,\omega '}^nf(x,y)&=\left[D_n(\omega _{\ab{x}-1},\omega '_{\ab{y}-1})
  \delta _{x,\omega _{\ab{x}}}\delta _{y,\omega '_{\ab{y}}}\right.\\
  &\quad -\left.D_n(\omega '_{\ab{x}-1},\omega _{\ab{y}-1})\delta _{x,\omega '_{\ab{x}}}\delta _{y,\omega _{\ab{y}}}\right]
  f(x,y)
\end{align*}
It is shown in \cite{gud122} that
\begin{equation}         
\label{eq44}
\rscript _{\omega ,\omega '}^n=\dscript _{\omega ,\omega '}^n+\tscript _{\omega ,\omega '}^n
\end{equation}
We call \eqref{eq44} the \textit{discrete Einstein equation} \cite{gud122, wal84}. In this sense, Einstein's equation always holds in this framework no matter what we have for the quantum dynamics $\rho _n$. One might argue that we obtained \eqref{eq44} just by definition. However, our derivation shows that $\rscript _{\omega ,\omega '}^n$ is a reasonable counterpart of the classical curvature tensor and $\dscript _{\omega ,\omega '}^n$ is a discrete counterpart of the metric tensor.

Equation~\eqref{eq44} does not give information about $D_n(x,y)$ and $D_n(\omega ,\omega ')$ (which after all, are what we wanted to find), but it may give useful indirect information. If we can find $D_n(\omega ,\omega ')$ such that the classical Einstein equation is an approximation to \eqref{eq44}, then this gives information about $D_n(\omega ,\omega ')$. Moreover, an important problem in discrete quantum gravity theory is how to test whether general relativity is a close approximation to the theory. Whether Einstein's equation is an approximation to \eqref{eq44} would provide such a test. In order to consider approximations by Einstein's equation, it will be necessary to let $n\to\infty$ in \eqref{eq44}. However, the convergence of the operators depends on $D_n$ and will be left for later investigations.

We now propose a definition for discrete geodesics. For $A\subseteq Q_n$ we define the $q$-measure
\begin{equation*}
\mu _n(A)=\mu _n\paren{\brac{\omega\in\Omega _n\colon\omega _{\ab{x}}=x\hbox{ for some }x\in A}}
\end{equation*}
and for $x,y\in Q_n$ let
\begin{equation*}
\mu _n(x\cap y)=\mu _n\paren{\brac{\omega\in\Omega _n\colon\omega _{\ab{x}}=x\hbox{ and }\omega _{\ab{y}}=y}}
\end{equation*}
For $x,y\in Q_n$ it is natural to define the \textit{conditional} $q$-\textit{measure}
\begin{equation*}
\mu _n(x\mid y)=\frac{\mu _n(x\cap y)}{\mu _n(y)}\quad\hbox{ if }\mu _n(y)\ne 0
\end{equation*}
and $\mu _n(x\mid y)\!=\!0$ if $\mu _n(y)=0$. For $\omega\in\Omega _n$ define the function
$\omegahat\colon Q_n\to\real ^+$ by
\begin{equation}         
\label{eq45}
\omegahat (x)=\mu _n(x\mid\omega _{\ab{x}-1})\delta _{x,\omega _{\ab{x}}}
\end{equation}
Of course, $\omegahat\in K_n$. We say that $\omega\in\Omega _n$ is a \textit{discrete geodesic} if there is an $a\in\real$ such that $\varbigtriangleup _\omega ^n\omegahat =a\omegahat$; that is $\omegahat$ is an eigenvector of
$\varbigtriangleup _\omega ^n$.

\begin{thm}       
\label{thm43}
An $n$-path $\omega\in\Omega _n$ is a discrete geodesic if and only if there exists a $c\in\real$ such that whenever
$\omega _{\ab{x}}=x$ for $\ab{x}\ge 3$ we have
\begin{equation}         
\label{eq46}
\mu _n(x\mid\omega _{\ab{x}-1})=c\mu _n(\omega _{\ab{x}-1}\mid\omega _{\ab{x}-2})
\end{equation}
\end{thm}
\begin{proof}
By definition, $\omega$ is a discrete geodesic if and only if there is an $a\in\real$ such that for all $x\in Q_n$ we have
\begin{equation*}
a\omegahat (x)=\varbigtriangleup _\omega ^n\omegahat (x)=\sqbrac{\omegahat (x)-\omegahat (\omega _{\ab{x}-1})}
  \delta _{x,\omega _{\ab{x}}}
\end{equation*}
Letting $c=1-a$ this last statement is equivalent to
\begin{equation*}
\omegahat (x)=c\omegahat (\omega _{\ab{x}-1})\delta _{x,\omega _{\ab{x}}}
\end{equation*}
Applying \eqref{eq45} we conclude that
\begin{equation*}
\mu _n(x\mid \omega _{\ab{x}-1})\delta _{x,\omega _{\ab{x}}}=c\mu _n(\omega _{\ab{x}-1}\mid\omega _{\ab{x}-2})
  \delta _{x,\omega _{\ab{x}}}
\end{equation*}
If $\omega _{\ab{x}}\ne x$, both sides of this equation vanish so the equation holds. If $\omega _{\ab{x}}=x$ we obtain \eqref{eq46}.
\end{proof}

This definition of a discrete geodesic is very restrictive and it seems desirable to have a more general concept. If
$\omega =\omega _j\omega _{j+1}\cdots\omega _n$ with $\omega _i\in\pscript _i$, $\omega _i\to\omega _{i+1}$,
$i=j,\ldots ,n-1$, we call $\omega$ an $n$-\textit{path starting} at $\omega _j$. Motivated by Theorem~\ref{thm43}, we say that $\omega =\omega _j\omega _{j+1}\cdots\omega _n$ is a \textit{discrete geodesic starting } at $\omega _j$ if $\omega$ is a maximal $n$-path starting at $\omega _j$ satisfying
\begin{equation*}
\mu _n(\omega _k\mid\omega _{k+1})=c\mu _n(\omega _{k-1}\mid\omega _{k-2}),\quad k=j+2,\ldots ,n
\end{equation*}
for some $c\in\real$.

\section{Amplitude Processes} 
Various constructions of a QSGP have been investigated \cite{gud111, gud112,gud121}. In this section we introduce a simple type of QSGP that we call an amplitude process. If nothing else, this might serve as a toy model for discrete quantum gravity.

For $x\in\pscript _n$, $y\in\pscript _{n+1}$ with $x\to y$, let $a(x\to y)\in\complex$ satisfy
\begin{equation}         
\label{eq51}
\sum\brac{a(x\to y)\colon y\in x\offspring}=1
\end{equation}
We call $a(x\to y)$ a \textit{transition amplitude} from $x$ to $y$. By convention we define $a(x\to y)=0$ if $x\not\to y$. For
$\omega =\omega _1\omega _2\cdots\omega _n\in\Omega _n$ we define the \textit{amplitude} of $\omega$ by
\begin{equation*}
a_n(\omega )=a(\omega _1\to\omega _2)a(\omega _2\to\omega _3)\cdots a(\omega _{n-1}\to\omega _n)
\end{equation*}
and we call the vector $a_n\in H_n$ an \textit{amplitude vector}. For $\omega ,\omega '\in\Omega _n$ define the
\textit{decoherence matrix} as
\begin{equation*}
D_n(\omega ,\omega ')=a_n(\omega )\overline{a_n(\omega ')}
\end{equation*}
Let $\rho _n$ be the operator on $H_n$ given by the matrix $D_n(\omega ,\omega ')$. We call the sequence of operators
$\rho _n$, $n=1,2,\ldots$, an \textit{amplitude process} (AP).

\begin{thm}       
\label{thm51}
An AP $\rho _n$ is a QSGP.
\end{thm}
\begin{proof}
It is clear that $\rho _n$ is a positive operator on $H_n$. Moreover, we have
\begin{align}         
\label{eq52}
\elbows{\rho _n1_n,1_n}&=\elbows{\sum _{\omega '\in\Omega _n}D_n(\omega ,\omega '),1_n}
  =\sum _{\omega ,\omega '\in\Omega _n}D_n(\omega ,\omega ')\notag\\
  &=\sum _{\omega ,\omega '\in\Omega _n}a_n(\omega )\overline{a_n(\omega ')}
  =\ab{\sum _{\omega\in\Omega _n}a_n(\omega )}^2
\end{align}
Applying \eqref{eq51} we obtain
\begin{align}         
\label{eq53}
\sum _{\omega\in\Omega _n}a_n(\omega )
  &=\sum a(\omega _1\to\omega _2)a(\omega _2\to\omega _3)\cdots a(\omega _{n-1}\to\omega _n)\notag\\
  &=\sum a(\omega _1\to\omega _2)\cdots a(\omega _{n-2}\to\omega _{n-1})
  \sum _{\omega _{n-1}\offspring}a(\omega _{n-1}\to\omega _n)\notag\\
  &=\sum a(\omega _1\to\omega _2)\cdots a(\omega _{n-2}\to\omega _{n-1})\notag\\
  &\qquad\vdots\notag\\
  &=\sum _{\omega _1\offspring}a(\omega _1\to\omega _2)=1
\end{align}
By \eqref{eq52} and \eqref{eq53} we conclude that $\elbows{\rho _n1_n,1_n}=1$. To show that $\rho _n$ is a consistent sequence, let $\omega ,\omega '\in\Omega _n$ with $\omega =\omega _1\omega _2\cdots\omega _n$,
$\omega '=\omega '_1\omega '_2\cdots\omega '_n$. By \eqref{eq51} we have
\begin{align}         
\label{eq54}
D_{n+1}&(\omega\offspring ,\omega '\offspring)
 =\elbows{\rho _{n+1}\chi _{\omega '\offspring},\chi _{\omega\offspring}}\notag\\
  &=\sum\brac{a_n(\omega )a(\omega _n\to x)\overline{a_n(\omega ')}\,\overline{a (\omega '_n\to y)}\colon
  \omega _n\to x,\omega '_n\to y}\notag\\
  &=a_n(\omega )\overline{a_n(\omega ')}\sum\brac{a(\omega _n\to x)\colon\omega _n\to x}
  \sum \brac{\overline{a(\omega '_n\to y)}\colon\omega '_n\to y}\notag\\
  &=a_n(\omega )\overline{a_n(\omega ')}=D_n(\omega ,\omega ')
\end{align}
For $A,B\in\ascript _n$, by \eqref{eq54} we have
\begin{align*}
D_{n+1}(A\offspring ,B\offspring )
  &=\sum\brac{D_{n+1}(\omega\offspring ,\omega '\offspring)\colon\omega\in A,\omega '\in B}\\
  &=\sum\brac{D_n(\omega ,\omega ')\colon\omega\in A,\omega '\in B}\\
  &=D_n(A,B)
\end{align*}
\end{proof}

Since the operator $\rho _n$ on $H_n$ has the form $\rho _n=\ket{a_n}\bra{a_n}$ we not only see that $\rho _n$ is a positive operator but that it has rank~$1$ with norm
\begin{equation*}
\|\rho _n\|=\|\ket{a_n}\bra{a_n}\|=\|a_n\|^2=\sum\ab{a_n(\omega )}^2=\rmtr (\rho _n)
\end{equation*}
The decoherence functional corresponding to $\rho _n$ becomes
\begin{align*}
D_n(A,B)&=\elbows{\rho _n\chi _B,\chi _A}=\elbows{\ket{a_n}\bra{a_n}\chi _B,\chi _A}\\
  &=\elbows{a_n,\chi _A}\elbows{\chi _B,a_n}
  =\sum _{\omega\in A}a_n(\omega )\sum _{\omega '\in B}\overline{a_n(\omega )}\\
  &=\sum\brac{D_n(\omega ,\omega ')\colon\omega\in A,\omega '\in B}
\end{align*}
for all $A,B\in\ascript _n$ which is what we expect. The corresponding $q$-measure is given by
\begin{equation}         
\label{eq55}
\mu _n(A)=D_n(A,A)=\ab{\elbows{a_n,\chi _A}}^2=\ab{\sum _{\omega\in A}a_n(\omega )}^2
\end{equation}
for all $A\in\ascript _n$. In particular, for $\omega\in\Omega _n$ we have
$\mu _n\paren{\brac{\omega}}=\ab{a_n(\omega )}^2$. We conclude that $A\in\bscript (\rho _n)$ if and only if
\begin{equation*}
\lim\mu _n(A^n)=\lim\ab{\sum _{\omega\in A^n}a_n(\omega )}^2
\end{equation*}
exists and is finite in which case $\mu (A)$ is the limit.

\section{Complex Percolation Process} 
This section introduces a particular type of AP that still has physical relevance. We use the notation
$y=x\shortuparrow a$ if
$x\to y$ and $y$ is obtained from $x$ by adjoining the maximal element $a$ to $x$. Let $r\in\complex$ with $r\ne 0,1$. For $x,y\in\pscript$ with $y=x\shortuparrow a$ define
\begin{equation}         
\label{eq61}
a(x\to y)=m(x\to y)r^p(1-r)^u
\end{equation}
where $p$ is the number of parents of $a$ and $u$ is the number of unrelated (non ancestors, not equal to $a$) elements of $a$.

\begin{thm}       
\label{thm61}
If $a(x\to y)$ is given by \eqref{eq61}, then $a(x\to y)$ satisfies \eqref{eq51} and hence is a transition amplitude.
\end{thm}
\begin{proof}
We prove the result by strong induction on $\ab{x}$. If $\ab{x}=1$, then $x=x_1$ and $(x\offspring )=\brac{x_2,x_3}$ in Figure~1. Hence,
\begin{equation*}
\sum _{x_1\offspring}a(x\to y)=a(x_1\to x_2)+a(x_1\to x_3)=r+(1-r)=1
\end{equation*}
If $\ab{x}=2$, then $x=x_2$ or $x=x_3$ and $(x_2\offspring )=\brac{x_4,x_5,x_6}$, $(x_3\offspring )=\brac{x_6,x_7,x_8}$ in Figure~1. Hence,
\begin{align*}
\sum _{x_2\offspring}a(x_2\to y)&=r+r(1-r)+(1-r)^2=1\\
\sum _{x_3\offspring}a(x_3\to y)&=2r(1-r)+r^2+(1-r)^2=1
\end{align*}
Now assume the result holds for $\ab{x}\le n$ and suppose that $\ab{x}=n+1\ge 3$. We have that $x=x'\shortuparrow a$ for some $x'\in\pscript _n$ and if $z\in x\offspring$, then $z=x\shortuparrow b$. Let
\begin{align*}
A&=\brac{z\in x\offspring\colon b\not> a}\\
B&=\brac{z\in x\offspring\colon b>a}\\
v&=\brac{c\in x\colon c\not\le a}\in\pscript
\end{align*}
We then have
\begin{align*}
\sum\brac{a(x\to z)\colon z\in A}&=(1-r)\sum _{x'\offspring}a(x'\to y')=1-r\\
\sum\brac{a(x\to z)\colon z\in B}&=r\sum _{v\offspring}a(v\to v')=r
\end{align*}
Hence,
\begin{equation*}
\sum _{x\offspring}a(x\to y)=(1-r)+r=1
\end{equation*}
This completes the induction proof.
\end{proof}

It follows from Theorems \ref{thm51} and \ref{thm61} that if $a(x\to y)$ is given by \eqref{eq61}, then the operator $\rho _n$ corresponding to the matrix $D_n(\omega ,\omega ')=a_n(\omega )\overline{a(\omega ')}$ forms an AP and hence a QSGP. We then call $\rho _n$ $n=1,2,\ldots$, a \textit{complex percolation process} (CPP) with \textit{percolation constant} $r$. The form of \eqref{eq61} was chosen because it conforms with general principles of causality and covariance \cite{rs00, vr06}.

As an illustration of a CPP suppose the percolation constant is
\begin{equation*}
r=\frac{1}{\,\sqrt{2}\,}\, e^{i\pi /4}=\frac{1}{2}+\frac{i}{2}
\end{equation*}
This example may have physical relevance because $r$ is the unique complex number satisfying $\ab{r}^2=\ab{1-r}^2=1/2$. Notice that $1-r=\overline{r}$ and we have
\begin{equation*}
a(x\to y)=m(x\to y)r^p\,\overline{r}^{\,u}=\frac{m(x\to y)}{2^{(p+u)/2}}\,e^{i(p-u)\pi /4}
\end{equation*}
For $n=2$, letting $\gamma _1=x_1x_2$, $\gamma _2 =x_1x_3$ from Figure~1, the amplitude vector becomes
\begin{equation*}
a_2=\frac{1}{\,\sqrt{2}\,}\,(e^{i\pi /4},e^{-i\pi /4})
\end{equation*}
and the decoherence matrix is
\begin{equation*}
D_2=\frac{1}{2}\left[\begin{matrix}\noalign{\smallskip}1&i\\-i&1\\\end{matrix}\right]
\end{equation*}
We then have that $\|\rho _2\|=\rmtr(\rho _2)=1$.

For further computations, it is useful to list the transition amplitudes for the causets of Figure~1.
\vglue 2pc

{\hskip - 4pc
\begin{tabular}{c|c|c|c|c|c|c|c}
$(i,j)$&$(1,2)$&$(1,3)$&$(2,4)$&$(2,5)$&$(2,6)$&$(3,7)$&$(3,8)$\\
\hline
$a(x_i\to x_j)$&$r$&$1-r$&$r$&$r(1-r)$&$2r(1-r)$&$r^2$&$(1-r)^2$\\
\end{tabular}}
\vskip 2pc
{\hskip - 4pc
\begin{tabular}{c|c|c|c|c|c|c|c}
$(4,9)$&$(4,10)$&$(4,11)$&$(4,14)$&$(5,11)$&$(5,12)$&$(5,13)$&$(5,15)$\\
\hline
$r$&$r(1-r)^2$&$r(1-r)^2$&$(1-r)^3$&$2r(1-r)$&$r(1-r)^2$&$r^2$&$(1-r)^3$\\
\end{tabular}}
\vskip 2pc
\begin{tabular}{c|c|c|c|c|c|c}
$(6,14)$&$(6,15)$&$(6,16)$&$(6,17)$&$(6,18)$&$(6,19)$&$(7,18)$\\
\hline
$r(1-r)$&$r(1-r)^2$&$r(1-r)^2$&$r^2$&$r^2(1-r)$&$(1-r)^3$&$2r(1-r)^2$\\
\end{tabular}
\vskip 2pc
\begin{tabular}{c|c|c|c|c|c|c}
$(7,20)$&$(7,21)$&$(7,22)$&$(8,19)$&$(8,22)$&$(8,23)$&$(8,24)$\\
\hline
$r$&$r^2(1-r)$&$(1-r)^3$&$3r(1-r)^2$&$3r^2(1-r)$&$r^3$&$(1-r)^3$\\
\noalign{\bigskip}
\multicolumn{7}{c}%
{\textbf{Table 1}}\\
\end{tabular}
\vskip 1.5pc

Letting $\gamma _1=x_1x_2x_4$, $\gamma _2=x_1x_2x_5$, $\gamma _3=x_1x_2x_6$,
$\gamma _4=x_1x_3x_6$, $\gamma _5=x_1x_3x_7$, $\gamma _6=x_1x_3x_8$ in Figure~1, we have that
$\Omega _6=\brac{\gamma _1,\ldots ,\gamma _6}$ with amplitude vector
\begin{equation*}
a_3=2^{-3/2}(\sqrt{2}\,i,e^{i\pi /4},e^{-i\pi /4},2e^{-i\pi /4},e^{i\pi /4}e^{-i3\pi /4})
\end{equation*}
We conclude that $\|\rho _3\|=\rmtr (\rho _3)=5/4$. The decoherence matrix can be computed from
\begin{equation*}
D_3(\omega ,\omega ')=\sqbrac{a_3(\omega )\,\overline{a_3(\omega ')}}
\end{equation*}
the following table gives an ordering of the paths $\gamma _j\in\Omega _4$.
\vskip 2pc

{\hskip - 4pc
\begin{tabular}{c|c|c|c|c|c|c|c}
$j$&$1$&$2$&$3$&$4$&$5$&$6$&$7$\\
\hline
$\gamma _j$&$x_1x_2x_4x_9$&$x_1x_2x_4x_{10}$&$x_1x_2x_4x_{11}$&$x_1x_2x_4x_{14}$
&$x_1x_2x_5x_{11}$&$x_1x_2x_5x_{12}$&$x_1x_2x_5x_{13}$\\
\end{tabular}|
\vskip 1.5pc
{\hskip - 4pc
\begin{tabular}{c|c|c|c|c|c|c}
$8$&$9$&$10$&$11$&$12$&$13$&$14$\\
\hline
$x_1x_2x_5x_{15}$&$x_1x_2x_6x_{14}$&$x_1x_2x_6x_{15}$&$x_1x_2x_6x_{16}$
&$x_1x_2x_6x_{17}$&$x_1x_2x_6x_{18}$&$x_1x_2x_6x_{19}$\\
\end{tabular}}
\vskip 1.5pc
{\hskip - 4pc
\begin{tabular}{c|c|c|c|c|c|c}
$15$&$16$&$17$&$18$&$19$&$20$&$21$\\
\hline
$x_1x_3x_6x_{14}$&$x_1x_3x_6x_{15}$&$x_1x_3x_6x_{16}$&$x_1x_3x_6x_{17}$
&$x_1x_3x_6x_{18}$&$x_1x_3x_6x_{19}$&$x_1x_3x_7x_{18}$\\
\end{tabular}}
\vskip 1.5pc
{\hskip - 4pc
\begin{tabular}{c|c|c|c|c|c|c}
$22$&$23$&$24$&$25$&$26$&$27$&$28$\\
\hline
$x_1x_3x_7x_{20}$&$x_1x_3x_7x_{21}$&$x_1x_3x_7x_{22}$&$x_1x_3x_8x_{19}$
&$x_1x_3x_8x_{22}$&$x_1x_3x_8x_{23}$&$x_1x_3x_8x_{24}$\\
\noalign{\bigskip}
\multicolumn{7}{c}%
{\textbf{Table 2}}\\
\end{tabular}}
\vskip 2pc

Then $\Omega _4=\brac{\gamma _1,\gamma _2,\ldots ,\gamma _{28}}$ with amplitude vector
\begin{align*}
a_4&=\tfrac{1}{8}(2^{3/2}e^{i3\pi/4},2i,\sqrt{2}\,e^{i\pi /4},\sqrt{2}\,e^{-i\pi/4},2^{3/2}e^{i\pi /4},1,2i,-i,\sqrt{2}\,e^{-i\pi /4},\\
  &\quad -i,-i,\sqrt{2}\,e^{i\pi /4},1,-1,2^{3/2}e^{-i\pi /4},-2i,-2i,2^{3/2}e^{i\pi /4},\\
  &\qquad 2,-2,2,2i,i,-i,-3,-3i,1,i)
\end{align*}
We then have $\|\rho _4\|=\rmtr (\rho _4)=25/16$. Although we have not been able to show this, we conjecture that in
general $\|\rho _n\|=(5/4)^{n-2}$.

We now compute some $q$-measures. For $\Omega _2=\brac{\gamma _1,\gamma _2}$ we have
$\mu _2(\gamma _1)=\mu _2(\gamma _2)=1/2$, $\mu _2(\Omega _2)=1$. Also, $\mu _2(x_2)=\mu _2(x_3)=1/2$,
$\mu _2(\pscript _2)=1$. In this case there is no interference and $\mu _2$ is a measure.

For $\Omega _3=\brac{\gamma _1,\ldots ,\gamma _6}$ we have
\begin{align*}
\mu _3(\gamma _1)&=1/4,\quad\mu _3(\gamma _2)=\mu _3(\gamma _3)=\mu (\gamma _5)=\mu _3(\gamma _6)=1/8\\
\mu _3(\gamma _4)&=1/2,\quad\mu _3(\Omega _3)=1
\end{align*}
Moreover, by \eqref{eq55} we have
\begin{align*}
\mu _3\paren{\brac{\gamma _1,\gamma _2}}&=\mu _3(\gamma _1)+\mu _3(\gamma _2)+2\rmre a_3(\gamma _1)
  \overline{a_3(\gamma _2)}\\
  &\qquad \tfrac{1}{4}+\tfrac{1}{8}+\tfrac{1}{4}\rmre\paren{\sqrt{2}\,ie^{-i\pi /4}}=\tfrac{5}{8}
\end{align*}
The $q$-measures of the other doubleton sets are computed in a similar way. These are summarized in Table~3.
\vskip 2pc

\begin{tabular}{c|c|c|c|c|c|c|c}
$(j,k)$&$(1,3)$&$(1,4)$&$(1,5)$&$(1,6)$&$(2,3)$&$(2,4)$&$(2,5)$\\
\hline
$\mu _3\paren{\brac{\gamma _j,\gamma _k}}$&$1/8$&$1/4$&$5/8$&$1/8$&$1/4$&$5/8$&$1/2$\\
\end{tabular}
\vskip 1.5pc
\begin{tabular}{c|c|c|c|c|c|c}
$(2,6)$&$(3,4)$&$(3,5)$&$(3,6)$&$(4,5)$&$(4,6)$&$(5,6)$\\
\hline
$0$&$9/8$&$1/4$&$1/4$&$5/8$&$5/8$&$0$\\
\noalign{\bigskip}
\multicolumn{7}{c}%
{\textbf{Table 3}}\\
\end{tabular}
\vglue 2pc

The other $q$-measure for $\Omega _3$ can be found using our previous results and grade-2 additivity. For example,
\begin{align*}
\mu _3\paren{\brac{\gamma _2,\gamma _5,\gamma _6}}
  &=\mu _3\paren{\brac{\gamma _2,\gamma _5}}+\mu _3\paren{\brac{\gamma _2,\gamma _6}}
  +\mu _3\paren{\brac{\gamma _5,\gamma _6}}\\
  &\quad -\mu _3(\gamma _2)-\mu _3(\gamma _5)-\mu _3(\gamma _6)=1/8
\end{align*}
In a similar way we obtain
\begin{align*}
\mu _3\paren{\brac{\gamma _1,\gamma _2,\gamma_3}}
  &=1/2,\quad\mu _3\paren{\brac{\gamma _1,\gamma _2,\gamma _6}}=1/4\\
\mu _3\paren{\brac{\gamma _1,\gamma _2,\gamma_4}}
  &=5/8,\quad\mu _3\paren{\brac{\gamma _2,\gamma _3,\gamma _4}}=5/4\\
\end{align*}
An example of a $4$-element set is given by
\begin{align*}
\mu _3\paren{\brac{\gamma _1,\gamma _2,\gamma _3,\gamma _4}}&=\mu _3\paren{\brac{\gamma _1,\gamma _2}}
  +\mu _3\paren{\brac{\gamma _1,\gamma _3}}+\mu _3\paren{\brac{\gamma _1,\gamma _4}}\\
  &\quad +\mu _3\paren{\brac{\gamma _2,\gamma _3}}+\mu _3\paren{\brac{\gamma _2,\gamma _4}}
  +\mu _3\paren{\brac{\gamma _3,\gamma _4}}\\
  &\quad -2\sqbrac{\mu _3(\gamma _1)+\mu _3(\gamma _2)+\mu _3(\gamma _3)+\mu _3(\gamma _4)}=1
\end{align*}

The $q$-measure of paths in $\Omega _4=\brac{\gamma _1,\ldots ,\gamma _{28}}$ are given in Table~4.
\vskip 2pc

{\hskip - 2pc
\begin{tabular}{c|c|c|c|c|c|c|c|c|c|c}
$j$&$1$&$2$&$3$&$4$&$5$&$6$&$7$&$8$&$9$&$10$\\
\hline
$\mu _4(\gamma _j)$&$1/8$&$1/16$&$1/32$&$1/32$&$1/8$&$1/64$&$1/16$&$1/64$&$1/32$&$1/64$\\
\end{tabular}|
\vskip 1.5pc
{\hskip - 2pc
\begin{tabular}{c|c|c|c|c|c|c|c|c|c}
$11$&$12$&$13$&$14$&$15$&$16$&$17$&$18$&$19$&$20$\\
\hline
$1/64$&$1/32$&$1/64$&$1/64$
&$1/8$&$1/16$&$1/16$&$1/8$&$1/16$&$1/16$\\
\end{tabular}}
\vskip 1.5pc
{\hskip - 2pc
\begin{tabular}{c|c|c|c|c|c|c|c}
$21$&$22$&$23$&$24$&$25$&$26$&$27$&$28$\\
\hline
$1/16$&$1/16$&$1/64$&$1/64$&$9/64$&$9/64$&$1/64$&$1/64$\\
\noalign{\bigskip}
\multicolumn{8}{c}%
{\textbf{Table 4}}\\
\end{tabular}}
\vskip 2pc

\noindent Instead of finding $\mu _4(A)$ for arbitrary $A\in\ascript _4$, we compute $\mu _4(x)$, $x\in\pscript _4$ which are given in Table~5.
\vskip 2pc

{\hskip - 2pc
\begin{tabular}{c|c|c|c|c|c|c|c|c|c|c}
$j$&$9$&$10$&$11$&$12$&$13$&$14$&$15$&$16$&$17$&$18$\\
\hline
$\mu _4(x_j)$&$1/8$&$1/16$&$9/32$&$1/16$&$1/16$&$1/4$&$1/4$&$9/64$&$9/32$&$25/64$\\
\noalign{\bigskip}
\multicolumn{11}{c}%
{\textbf{Table 5}}\\
\end{tabular}}
\vskip 2pc

We now briefly consider the $q$-measure of some sets in $\bscript (\rho _n)\smallsetminus\cscript (\Omega )$. If
$\omega =\omega _1\omega _2\cdots\in\Omega$, then $\brac{\omega}\notin\cscript (\Omega )$ and
$\brac{\omega}^n=\brac{\omega _1\omega _2\cdots\omega _n}$. We define the \textit{multiplicity} $m(\omega )$ by
\begin{equation*}
m(\omega )=\prod _{j=1}^\infty m(\omega _j\to\omega _{j+1}
\end{equation*}
As suggested by Figure~1, most $\omega\in\Omega$ have finite multiplicity, although there are a few with
$m(\omega )=\infty$. If $m(\omega )<\infty$ then it is easy to verify that $\brac{\omega}\in\bscript (\rho _n)$ and
\begin{equation*}
\mu\paren{\brac{\omega}}=\lim _{n\to\infty}\ab{a_n(\omega _1\omega _2\cdots\omega _n)}^2=0
\end{equation*}
Moreover, it follows by \eqref{eq51} that if $m(\omega )<\infty$, then $\brac{\omega}'\in\bscript (\rho _n)$ and
\begin{equation*}
\mu\paren{\brac{\omega}'}=\lim _{n\to\infty}\ab{1-a_n(\omega _1\omega _2\cdots\omega _n)}^2=1
\end{equation*}
It would be interesting to investigate whether there exist $\omega\in\Omega$ with $m(\omega )=\infty$ and
$\omega\in\bscript (\rho _n)$.

Finally, we briefly discuss discrete geodesics. Let $\rho _n$ be an AP. If $x\to y$ with $x\in\pscript _n$, then
\begin{equation*}
\mu _n(x\cap y)=\mu _n\paren{\brac{\omega\colon\omega =\omega _1\omega _2\cdots xy}}
  =\mu _n(x)\ab{a(x\to y)}^2
\end{equation*}
It follows that $\mu _n(y\mid x)=\ab{a(x\to y)}^2$. We conclude that $\omega =\omega _j\omega _{j+1}\cdots\omega _n$ is a discrete geodesic starting with $\omega _j$ if and only if $\omega$ is maximal and satisfies
\begin{equation}         
\label{eq62}
\ab{a(\omega _{k-1}\to\omega _k)}=c\ab{a(\omega _{k-2}\to\omega _{k-1}}
\end{equation}
$k=j+2,\ldots ,n$, for some $c\in\real$. In particular, suppose $\rho _n$ is a CPP with percolation constant $r$. Letting $n=4$ and employing the notation of Figure~1, the only discrete geodesics starting at $x_1$ are $x_1x_2x_4x_9$ and
$x_1x_3x_8x_{24}$. The constant $c$ in \eqref{eq62} for $x_1x_2x_4x_9$ is $1$ and for $x_1x_2x_8x_{24}$ is
$\ab{1-r}$. A generic causet is contained in at least two discrete geodesics with $c=1$ or $c=\ab{1-r}$. However, there are exceptional causets that are contained in only one discrete geodesic. It would be interesting to classify causets according to their geodesic structure.


\begin{thebibliography}{99}
\bibitem{blms87}L.~Bombelli, J.~Lee, D.~Meyer and R.~Sorkin, 
Spacetime as a casual set, \textit{Phys. Rev. Lett.} \textbf{59} (1987), 521--524.
\bibitem{gud111}S.~Gudder, Discrete quantum gravity, arXiv: gr-qc 1108.2296 (2011).
\bibitem{gud112}S.~Gudder, Models for discrete quantum gravity, arXiv: gr-qc 1108.6036 (2011).
\bibitem{gud121}S.~Gudder, A matter of matter and antimatter, arXiv: gr-qc 1204.3346 (2012).
\bibitem{gud122}S.~Gudder, An Einstein equation for discrete quantum gravity, arXiv: gr-qc 1204.4506 (2012).
\bibitem{rs00}D.~Rideout and R.~Sorkin, A classical sequential growth dynamics for causal sets,
\textit{Phys. Rev. D} \textbf{61} (2000), 024002.
\bibitem{sor94}R.~Sorkin, 
Quantum mechanics as quantum measure theory, \textit{Mod. Phys. Letts.~A} \textbf{9} (1994), 3119--3127.
\bibitem{sor03}R.~Sorkin, 
Causal sets: discrete gravity, arXiv: gr-qc 0309009 (2003).
\bibitem{sur11}S.~Surya, 
Directions in causal set quantum gravity, arXiv: gr-qc 1103.6272 (2011).
\bibitem{vr06}M.~Varadarajan and D.~Rideout,
A general solution for classical sequential growth dynamics of causal sets,
\textit{Phys. Rev. D} \textbf{73} (2006), 104021.
\bibitem{wal84}R.~Wald, General Relativity, \textit{University of Chicago Press}, Chicago 1984.

\end{thebibliography}
\end{document}